\documentclass[a4paper,12pt]{article}
\usepackage{amssymb, amsmath, amsthm, cite,mathrsfs}
\numberwithin{equation}{section}

\allowdisplaybreaks[3]

\newcommand\R{{\mathbb{R}}}

\newcommand\N{{\mathbb{N}}}

\newcommand\B{{\mathcal{B}}}
\newcommand\M{{\mathcal{M}}}
\newcommand\D{{\mathcal{D}}}
\newcommand\X{{{\mathbb{R}}^d}}
\newcommand\1{{\rm 1\kern -.3600em 1}}
\newcommand\KK{{\mathrm{K}}}
\newcommand\K{{\mathscr{K}}}
\renewcommand\L{{\mathscr{L}}}
\newcommand\Bbs{{B_\mathrm{bs}(\Ga_0^2)}}
\newcommand{\ga}{{\gamma}}
\newcommand{\Ga}{{\Gamma}}
\newcommand{\la}{{\lambda}}
\newcommand{\La}{{\Lambda}}

\newcommand\cnt[2]{\text{\setbox2=\hbox{#1}\rlap{\hbox to \wd2{\hfil#2\hfil}}\box2}}
\newcommand\Star{\mathbin{\cnt{$\bigcirc$}{$\star$}}}

\newtheorem{theorem}{Theorem}[section]
\newtheorem{lemma}[theorem]{Lemma}
\newtheorem{proposition}[theorem]{Proposition}
\newtheorem{remark}[theorem]{Remark}

\begin{document}
\author{\textbf{Dmitri L.~Finkelshtein} \\
{\small Institute of Mathematics, Ukrainian National Academy of Sciences, 
01601 Kiev, Ukraine}\\
{\small fdl@imath.kiev.ua} \and
\textbf{Yuri G.~Kondratiev} \\
{\small Fakult\"at f\"ur Mathematik, Universit\"at Bielefeld,
D 33615 Bielefeld, Germany}\\
{\small Forschungszentrum BiBoS, Universit\"at Bielefeld, 
D 33615 Bielefeld, Germany}\\
{\small kondrat@mathematik.uni-bielefeld.de} \and 
\textbf{Maria Jo\~{a}o Oliveira} \\
{\small Universidade Aberta, P 1269-001 Lisbon, Portugal}\\
{\small CMAF, University of Lisbon, P 1649-003 Lisbon, Portugal}\\ 
{\small oliveira@cii.fc.ul.pt}}

\title{Markov evolutions and hierarchical equations in the continuum: II.~Multicomponent systems}

\date{}
\maketitle

\begin{abstract}
General birth-and-death as well as hopping stochastic dynamics of
infinite multicomponent particle systems in the continuum are
considered. We derive the corresponding evolution equations for
quasi-observables and correlation functions. We also present sufficient 
conditions that allows us to consider these equations on suitable Banach 
spaces.
\end{abstract}


\medskip


\newpage

\section{Introduction}

Spatial Markov processes in $\X$ can be described as stochastic evolutions of 
locally finite configurations. From this standpoint, two important classes of 
stochastic dynamics are represented by birth-and-death and hopping Markov 
processes on the configuration space $\Gamma$ over $\X$,
\[
\Ga :=\bigl\{ \ga  \subset \X  : | \ga\cap\Lambda| <\infty, \hbox{ for
every compact } \Lambda\subset \X \bigr\}.
\]
These are processes where randomly, at each random moment of time, particles 
(or individuals) disappear and new particles appear or, in the case of hopping 
particle systems, particles hop over the space $\X$, according to rates which 
in both cases depend on the configuration of the whole system at that time. 
However, both cases concern only one type of particles. 

Motivated by concrete ecological models \cite{CFM08,DM10,FFK07}, 
socio-economics models or even mathematical physics problems, e.g., the Potts 
model \cite{GH96,GMRZ06,KZ07}, in this work we extend these two classes of 
stochastic dynamics to Markov stochastic evolutions of different particle 
types. For simplicity of notation, we just present this extension for 
two particle types. A similar procedure applies to $n>2$ particle types, but 
with a more cumbersome notation.    

Since two particles cannot be located at the same position, the natural phase 
space is a subset of the direct product of two copies of the space $\Ga$, 
$\Ga^+$ and $\Ga^-$, namely,
\[
\Ga^2:=\bigl\{(\ga ^+,\ga ^-)\in\Ga^+\times\Ga^-: \ga ^+\cap\ga
^-=\emptyset\bigr\}.
\]
Given a configuration $(\ga^+,\ga^-)\in\Ga^2$, the aforementioned fields of 
applications suggest that, according to certain rates of probability, at 
each random moment of time several random phenomena may occur:\footnote{Here 
and below, for simplicity of notation, we have just written $x, y$ instead of 
$\{x\}, \{y\}$, respectively.} 
\begin{description}
\item[\qquad \rm Death of a $+$-particle:]
$(\ga^+,\ga^-)\longmapsto (\ga^+\setminus x,\ga^-)$, $x\in\ga^+$;
\item[\qquad \rm Birth of a new $+$-particle:] $(\ga^+,\ga^-)\longmapsto (\ga^+\cup x,\ga^-)$, $x\in(\X\setminus\ga^+)\setminus \ga^-$;
\item[\qquad \rm Hop of a $+$-particle to a free site:] $$(\ga^+,\ga^-)\longmapsto (\ga^+\setminus x\cup y,\ga^-), \qquad x\in\ga^+,\ y\in(\X\setminus\ga^+)\setminus \ga^-;$$
\item[\qquad \rm Hop of a $+$-particle flipping the mark to $-$:]  $$(\ga^+,\ga^-)\longmapsto (\ga^+\setminus x,\ga^-\cup y), \qquad x\in\ga^+,\ y\in(\X\setminus\ga^+)\setminus \ga^-;$$
\item[\qquad \rm Flip the mark $+$ to $-$, keeping the site:]  $$(\ga^+,\ga^-)\longmapsto (\ga^+\setminus x,\ga^-\cup x), \qquad x\in\ga^+.$$
\end{description}
Similar events naturally may occur with $-$-particles. In other words, besides 
the natural complexity imposed by the existence of different particle types, 
the treatment of multicomponent particle systems also deals with a higher 
number of possible random phenomena. 

Heuristically, the stochastic dynamics of a multicomponent particle system is 
described through a Markov generator $L$ defined according to the 
aforementioned elementary random phenomena and corresponding rates. The time 
evolution of states (that is, probability measures on $\Ga^2$) in the weak 
form may be formulated by means of the following initial value problems
\begin{equation}\label{Eq}
\frac{d}{d t}\langle F,\mu_t\rangle=\langle LF,\mu_t\rangle, \qquad
\mu_t\bigr|_{t=0}=\mu_0,
\end{equation}
for a wide class of functions $F$ on $\Ga^2$ (where 
$\langle\cdot,\cdot\rangle$ is the usual dual pairing between functions and 
measures on $\Ga^2$). For the study of \eqref{Eq}, we may consider the 
corresponding time evolution equations for correlation functionals 
(factorial moments) $k_t$ corresponding to the measures $\mu_t$. These are 
equations having a hierarchical structure similar to the well-known 
BBGKY-hierarchy for the Hamiltonian dynamics. However, in applications, 
frequently correlation functionals are not integrable, being a technical 
difficulty to proceed this study, even in a weak sense (corresponding to 
\eqref{Eq}). Having in mind the construction of a weak solution, we then 
analyze the (pre-)dual problem, that is, the so-called time evolution of 
quasi-observables. These are functions which naturally can be considered in 
proper spaces of integrable functions, allowing then to overtake the 
technical difficulties pointed out. Furthermore, the evolution equation for 
quasi-observables still has hierarchical structure.

For further developments and applications, in this work explicit formulas 
for the aforementioned hierarchical equations of general birth-and-death,  
hopping, and flipping multicomponent particle systems are derived. For the 
one-component case, a similar scheme has been proposed in \cite{FKO05} and 
explicit forms for corresponding hierarchical equations have been presented 
therein. Within this setting, problems concerning one-component hierarchies 
and many applications were exposed, e.g., in \cite{FKKoz2011, FKK09, FKK11, 
FKK10,FKK10a,FKKZ2010,KKM2008,KKP06,KoKtZh06}. Naturally, due to the complexity 
mentioned above, one cannot infer from the one-component case corresponding 
results for multicomponent systems. Motivated by recent applications, in this 
work we slightly change the procedure used in \cite{FKO05}, which for 
one-component birth-and-death models is used in \cite{FKK11}.    
This change allows, in particular, to wide the class of rates. Sufficient 
conditions on the rates to give rise to linear operators 
on suitable Banach spaces and concrete examples of rates are analyzed as well.
 
\section{Markov evolutions in multicomponent configuration spaces}\label{Section2}

\subsection{One-component configuration spaces}\label{Subsection21}

The configuration space $\Ga :=\Ga _{\X }$ over $\X $, $d\in\N$, is
defined as the set of all locally finite subsets of $\X $ (that is, 
configurations),
\[
\Ga :=\bigl\{ \ga  \subset \X  : | \ga _\Lambda| <\infty, \hbox{ for
every compact } \Lambda\subset \X \bigr\} ,
\]
where $\left| \cdot \right|$ denotes the cardinality of a set and
$\ga _\Lambda := \ga  \cap \Lambda$. We identify each $\ga  \in \Ga
$ with the non-negative Radon measure $\sum_{x\in \ga }\delta_x\in
\M (\X )$, where $\delta_x$ is the Dirac measure with unit mass at
$x$, $\sum_{x\in\emptyset}\delta_x$ is, by definition, the zero
measure, and $\M (\X )$ denotes the space of all non-negative Radon
measures on the Borel $\sigma$-algebra $\B (\X )$. This
identification allows to endow $\Ga $ with the topology induced by
the vague topology on $\M (\X )$, that is, the weakest topology on
$\Ga$ with respect to which all mappings $\Ga \ni \ga  \mapsto
\sum_{x\in \ga  }f(x)$, $f\in C_c(\X )$, are continuous. Here
$C_c(\X )$ denotes the set of all continuous functions on $\X $ with
compact support. We denote by $\B (\Ga )$ the corresponding Borel
$\sigma$-algebra on $\Ga$.

Let us now consider the space of finite configurations
$$\Ga_0 := \bigsqcup_{n=0}^\infty \Ga^{(n)},$$
where $\Ga^{(n)} :=  \{ \ga \in \Ga: \vert \ga \vert = n\}$ for
$n\in \N$ and $\Ga^{(0)} := \{\emptyset\}$. For $n\in \N$, there is
a natural bijection between the space $\Ga^{(n)}$ and the
symmetrization $\widetilde{(\X )^n} \diagup S_n$ of the set
$\widetilde{(\X )^n}:= \{(x_1,...,x_n)\in (\X )^n: x_i\not= x_j
\hbox{ if } i\not= j\}$ under the permutation group $S_n$ over
$\{1,...,n\}$ acting on $\widetilde{(\X )^n}$ by permuting the
coordinate indexes. This bijection induces a metrizable topology on
$\Ga^{(n)}$, and we endow $\Ga_0$ with the metrizable topology of
disjoint union of topological spaces. We denote the corresponding Borel 
$\sigma$-algebras on $\Ga^{(n)}$ and $\Ga_0$ by $\B (\Ga^{(n)})$ and 
$\B(\Ga_0)$, respectively.

We proceed to consider the $K$-transform \cite{Le72,Le75a, Le75b,KoKu99}. Let 
$\B _c(\X )$ denote the set of all bounded Borel sets in $\X $,
and for each $\Lambda\in \B _c(\X )$ let $\Ga_\Lambda := \{\eta\in
\Ga: \eta\subset \Lambda\}$. Evidently $\Ga_\Lambda =
\bigsqcup_{n=0}^\infty \Ga_\Lambda^{(n)}$, where
$\Ga_\Lambda^{(n)}:= \Ga_\Lambda \cap \Ga^{(n)}$, $n\in
\N_0:=\N\cup\{0\}$, leading to a
 situation similar to the one for $\Ga_0$, described above.
We endow $\Ga_\Lambda$ with the topology of the disjoint union of
topological spaces and with the corresponding Borel $\sigma$-algebra
$\B (\Ga_\Lambda)$. To define the $K$-transform, among the functions
defined on $\Ga_0$ we distinguish the bounded $\B
(\Ga_0)$-measurable functions $G$ with bounded support, i.e.,
$G\!\!\upharpoonright _{\Ga_0\setminus \left(\bigsqcup_{n=0}^N\Ga
_\Lambda ^{(n)}\right)}\equiv 0$ for some $N\in\N_0$, $\Lambda \in
\B _c(\X )$. We denote the space of all such functions $G$ by
$B_{\mathrm{bs}}(\Ga_0)$. Given a $G\in B_{\mathrm{bs}}(\Ga_0)$, the 
$K$-transform of $G$ is a mapping $KG:\Ga\to\R$ defined at each
$\ga \in\Ga$ by
\begin{equation}
(KG)(\ga  ):=\sum_{\substack{\eta \subset \ga \\  \vert\eta\vert <
\infty} } G(\eta ). \label{Eq2.9}
\end{equation}
Note that for each function $G\in B_{\mathrm{bs}}(\Ga_0)$ the sum
in \eqref{Eq2.9} has only a finite number of summands different from
zero, and thus $KG$ is a well-defined function on $\Ga$. Moreover,
if $G$ has support described as before, then the restriction
$(KG)\!\!\upharpoonright_{\Ga _\Lambda}$ is a $\B
(\Ga_\Lambda)$-measurable function and $(KG)(\ga
)=(KG)\!\!\upharpoonright _{\Ga _\Lambda }\!\!(\ga _\Lambda)$ for
all $\ga \in\Ga$. That is, $KG$ is a cylinder function. In addition,
for each constant $C\geq\vert G\vert$ one finds $\vert(KG)(\ga )\vert\leq
C(1+\vert\ga _\Lambda\vert)^N$ for all $\ga \in\Ga$. As a result,
besides the cylindricity property, $KG$ is also polynomially
bounded.

It has been shown in \cite{KoKu99} that $K:
B_{\mathrm{bs}}(\Ga_0)\rightarrow K(B_{\mathrm{bs}}(\Ga_0))$ is a
linear isomorphism whose inverse mapping is defined by
\begin{equation*}
\left( K^{-1}F\right) (\eta ):=\sum_{\xi \subset \eta }(-1)^{|\eta
\backslash \xi |}F(\xi ),\quad \eta \in \Ga _0.
\end{equation*}

\subsection{Multicomponent configuration spaces}\label{Subsection-new}

The previous definitions naturally extend to any $n$-component configuration 
spaces. For simplicity of notation, we just present the extension for $n=2$.
A similar procedure is used for $n>2$, but with a more cumbersome notation.

Given two copies of the space $\Ga$, denoted by $\Ga^+$ and $\Ga^-$,
let
\[
\Ga^2:=\bigl\{(\ga ^+,\ga ^-)\in\Ga^+\times\Ga^-: \ga ^+\cap\ga
^-=\emptyset\bigr\}.
\]
Concerning the elements in $\Ga^2$, we observe they may be regarded
as marked one-configurations for the space of marks $\{+,-\}$
(spins). Similarly, given two copies of the space $\Ga_0$, $\Ga_0^+$ 
and $\Ga_0^-$, we consider the space
\[
\Ga^2_0:=\bigl\{(\eta^+,\eta^-)\in\Ga_0^+\times\Ga_0^-:
\eta^+\cap\eta^-=\emptyset\bigr\}.
\]

We endow $\Ga^2$ and $\Ga_0^2$ with the topology induced by the
product of the topological spaces $\Ga^+\times\Ga^-$ and
$\Ga_0^+\times\Ga_0^-$, respectively, and with the corresponding
Borel $\sigma$-algebras, denoted by $\B (\Ga^2)$ and $\B (\Ga^2_0)$.
Thus, a bounded $\B (\Ga^2_0)$-measurable function $G:\Ga^2_0\to\R$
has bounded support ($G\in \Bbs $, for short) whenever
$G\!\!\upharpoonright _{\Ga_0^2\backslash
\left(\bigsqcup_{n=0}^{N^+}\Ga
_{\Lambda^+}^{(n)}\times\bigsqcup_{n=0}^{N^-}\Ga
_{\Lambda^-}^{(n)}\right)}\equiv 0$ for some $N^+,N^-\in\N_0$,
$\Lambda^+,\Lambda^-\in \B _c(\X )$. In this way, given a function
$G\in \Bbs $, the mapping $\KK G$ defined at each $\ga =(\ga ^+,\ga
^-)\in\Ga^2$ by
\begin{equation}
(\KK G)(\ga  ):= \sum_{\substack{\eta^+\subset\ga ^+ \\
\vert\eta^+\vert < \infty}} \sum_{\substack{\eta^-\subset\ga ^-\\
\vert\eta^-\vert < \infty}} G(\eta^+,\eta^-) \label{M1}
\end{equation}
is a well-defined function on $\Ga^2$. For this verification, as
well as for other forthcoming ones, let us observe that given the
unit operator $I^\pm$ on functions on $\Ga^\pm$ (and thus, on
$\Ga_0^\pm$) and the operators defined on functions on $\Ga_0^2$ by
$K^+:=K\otimes I^-$, $K^-:=I^+\otimes K$ one may write,
equivalently to \eqref{M1},
\begin{equation}
\KK =K^+K^-=K^-K^+. \label{M2}
\end{equation}
We call the mapping $\KK G:\Ga^2\to\R$ the $\KK $-transform of $G$.

Either directly from definition \eqref{M1} or from \eqref{M2}, it is
clear that given a $G\in \Bbs $ described as before, the $\KK G$ is
a polynomially bounded cylinder function such that $(\KK G)(\ga
^+,\ga ^-)=(\KK G)(\ga ^+_{\Lambda^+},\ga ^-_{\Lambda^-})$ for all
$(\ga ^+,\ga ^-)\in\Ga^2$ and, for each constant $C\geq \vert G\vert$,
\[
\vert(\KK G)(\ga^+,\ga^-)\vert\leq C(1+\vert\ga
^+_{\Lambda^+}\vert)^{N^+}(1+\vert\ga
^-_{\Lambda^-}\vert)^{N^-},\quad (\ga ^+,\ga ^-)\in\Ga^2.
\]
Moreover, $\KK :\Bbs\rightarrow\mathcal{FP}(\Ga^2):=\KK (\Bbs )$ is a linear 
and positivity preserving isomorphism whose inverse mapping is defined by
\begin{equation}\label{inverseK2}
\left(\KK ^{-1}F\right) (\eta^+,\eta^-):=
\sum_{\xi^+\subset\eta^+}\sum_{\xi^-\subset\eta^-}
(-1)^{|\eta^+\backslash\xi^+|+|\eta^-\backslash\xi^-|}F(\xi^+,\xi^-),
\end{equation}
for all $(\eta^+,\eta^-)\in\Ga_0^2$.

\begin{remark}\label{remark}
Given any $\mathcal{B}(\Ga^2)$-measurable function $F$, observe that the 
right-hand side of \eqref{inverseK2} is also well-defined for 
$F\!\!\upharpoonright_{\Ga _0^2}$. In this case, since there will be no 
risk of confusion, we will denote the right-hand side of \eqref{inverseK2} by 
$\KK^{-1}F$.
\end{remark}

Let $\M _{\mathrm{fm}}^1(\Ga^2)$ denote the set of all probability
measures $\mu$ on $(\Ga^2,\B (\Ga^2))$ with finite local moments of
all orders, i.e.,
\begin{equation}
\int_{\Ga^2} d\mu(\ga^+,\ga^-)\, |\ga ^+_\Lambda|^n|\ga
^-_\Lambda|^n<\infty\quad \mathrm{for\,\,all}\,\,n\in\N
\mathrm{\,\,and\,\,all\,\,} \Lambda \in \B _c(\X ).\label{Dima5}
\end{equation}
Given a $\mu\in\M _{\mathrm{fm}}^1(\Ga^2)$, the so-called
correlation measure $\rho_\mu$ corresponding to $\mu$ is a measure
on $(\Ga _0^2,\B (\Ga_0^2))$ defined for all $G\in \Bbs $ by
\begin{equation}
\int_{\Ga_0^2}d\rho_\mu(\eta^+,\eta^-)\,G(\eta^+,\eta^-)=\int_{\Ga^2} d\mu(\ga^+,\ga^-)\,
\left(\KK G\right) (\ga^+,\ga^-).  \label{Eq2.16}
\end{equation}
Note that under these assumptions $\KK \left|G\right|$ is
$\mu$-integrable, and thus, \eqref{Eq2.16} is well-defined. In terms
of correlation measures, this means that $\Bbs \subset
L^1(\Ga_0^2,\rho_\mu)$. Actually, $\Bbs $ is dense in
$L^1(\Ga_0^2,\rho_\mu)$. Moreover, still by \eqref{Eq2.16}, on $\Bbs
$ the inequality $\Vert\KK G\Vert_{L^1(\Ga^2,\mu)}\leq \Vert
G\Vert_{L^1(\Ga_0^2,\rho_\mu)}$ holds, allowing an extension of the
$\KK $-transform to a bounded linear operator $\KK
:L^1(\Ga_0^2,\rho_\mu)\to L^1(\Ga^2,\mu)$ in such a way that
equality \eqref{Eq2.16} still holds for any $G\in
L^1(\Ga_0^2,\rho_\mu)$. For the extended operator the explicit form
\eqref{Eq2.9} still holds, now $\mu$-a.e.

Just to conclude this part, let us observe that in terms of correlation
measures property \eqref{Dima5} means that $\rho_\mu$ is
locally finite, that is, $\rho_\mu((\Ga _\Lambda ^{(n)}\times\Ga
_\Lambda ^{(m)})\cap\Ga_0^2)<\infty$ for all $n,m\in\N_0$ and all
$\Lambda\in\B _c(\X )$. 

\medskip

\noindent
\textbf{Poisson and Lebesgue-Poisson measures.} Given a constant $z>0$, 
let $\lambda_z$ be the Lebesgue--Poisson measure on $(\Ga_0,\B (\Ga_0))$,
\begin{equation}\label{LPm}
\lambda_z:=\sum_{n=0}^\infty \frac{z^n}{n!} m^{(n)},
\end{equation}
where each $m^{(n)}$, $n\in \N$, is the image measure on $\Ga^{(n)}$
of the product measure $dx_1...dx_n$ under the mapping
$\widetilde{(\X )^n}\ni (x_1,...,x_n)\mapsto\{x_1,...,x_n\}\in
\Ga^{(n)}$. For $n=0$ one sets $m^{(0)}(\{\emptyset\}):=1$. The product measure
$\lambda^2_z:=\lambda_z\otimes\lambda_z$ on $(\Ga_0^2,\B (\Ga_0^2))$
is the correlation measure corresponding to the product measure
$\pi_z\otimes\pi_z$ of the Poisson measure $\pi_z$ on $(\Ga,\B
(\Ga))$ with intensity $zdx$, that is, the probability measure defined on 
$(\Gamma,\mathcal{B}(\Gamma))$ by
\[
\int_\Gamma d\pi_z(\gamma)\,\exp \left( \sum_{x\in \gamma }\varphi (x)\right)
=\exp \left( z\int_{\R^d}dx\,\left( e^{\varphi (x)}-1\right)\right)
\]
for all smooth functions $\varphi$ on $\R^d$ with compact support.

If a correlation measure $\rho_\mu$ is absolutely
continuous with respect to the Lebesgue--Poisson measure
$\lambda^2:=\lambda_1^2$, the Radon--Nikodym derivative
$k_\mu:=\dfrac{d\rho_\mu}{d\lambda^2}$ is called the correlation
functional corresponding to $\mu$. Sufficient conditions for the existence of 
correlation functionals may be found e.g.~in \cite{F07}.

Technically, the next statement will be useful. It is an extension to the
multicomponent case of an integration result over $\Gamma_0$ 
(see e.g.~\cite{FicFre91,KMZ2004,Ru69}). 

\begin{lemma}
\label{Lmm2} The following equality holds
\begin{align}\label{MI}
&\int_{\Ga^2_0}d\lambda^2(\eta^+,\eta^-)\,
\sum_{\substack{\xi^+\subset\eta^+\\
\xi^-\subset\eta^-}}H(\eta^+,\eta^-,\xi^+,\xi^-)\\
=&\int_{\Ga_0^2}d\lambda^2(\eta^+,\eta^-)\int_{\Ga_0^2}d\lambda^2(\xi^+,\xi^-)\,
H(\eta^+\cup\xi^+,\eta^-\cup\xi^-,\xi^+,\xi^-)\nonumber
\end{align}
for  all measurable functions $H:\Ga^2_0\times  \Ga^2_0\to\R$ with respect to 
which at least one side of equality \eqref{MI} is finite for $|H|$.
\end{lemma}

\medskip

\noindent
\textbf{Algebraic properties.} The extension to functions defined on $\Ga_0^2$ 
of the $\star$-convolution introduced in \cite{KoKu99} for functions
defined on $\Ga_0$ has very similar properties. Given $G_1$ and
$G_2$ two $\B (\Ga_0^2)$-measurable functions we define the
$\Star$-convolution between $G_1$ and $G_2$ by
\begin{align}\nonumber
&(G_1\Star
G_2)(\eta^+,\eta^-)\\:=&\sum_{\substack{(\eta^+_1,\eta^+_2,\eta^+_3)\in\mathcal{P}_3(\eta^+)\\(\eta^-_1,\eta^-_2,\eta^-_3)\in\mathcal{P}_3(\eta^-)}}
G_1(\eta^+_1\cup\eta^+_2,\eta^-_1\cup\eta^-_2)G_2(\eta^+_2\cup\eta^+_3,\eta^-_2\cup\eta^-_3)\nonumber\\
=&\sum_{\substack{\xi^+\subset\eta^+\\ \xi^-\subset\eta^-}}G_1(\xi^+,\xi^-)
\sum_{\substack{\zeta^+\subset\xi^+\\ \zeta^-\subset\xi^-}}
G_2((\eta^+\setminus\xi^+)\cup\zeta^+,(\eta^-\setminus\xi^-)\cup\zeta^-),\label{dd}
\end{align}
where $\mathcal{P}_3(\eta^\pm)$ denotes the set of all partitions of
$\eta^\pm$ in three parts which may be empty. It is straightforward
to verify that the space of all $\B (\Ga_0^2)$-measurable functions
endowed with this product has the structure of a commutative algebra
with unit element $0^{|\eta^+|}0^{|\eta^-|}$. Furthermore, for each $G_1,
G_2\in \Bbs $ we have $G_1\Star G_2\in \Bbs $, and
\begin{equation*}
\KK \left( G_1\Star G_2\right) =\left(\KK G_1\right)\cdot\left(\KK
G_2\right). 
\end{equation*}

From definition \eqref{dd} it follows that for any 
$\mathcal{B}(\Ga^2)$-measurable functions $F_1,F_2$ such that 
$F_1\!\!\upharpoonright_{\Ga _0^2}, F_2\!\!\upharpoonright_{\Ga _0^2}$ are 
$\mathcal{B}(\Ga^2_0)$-measurable we have (cf.~Remark \ref{remark})
\begin{equation}\label{combin}
(\KK ^{-1}F_1)\Star(\KK ^{-1}F_2)=\KK ^{-1}(F_1F_2).
\end{equation}

\subsection{Markov generators and related evolution equations}\label{Subsection22}

Heuristically, the stochastic evolution of an infinite two-component particle 
system is described by a Markov process on $\Ga^2$, which is determined by a 
Markov generator $L$ defined on a proper space of functions on $\Ga^2$. If
such a Markov process exists, then it provides a solution to the (backward) 
Kolmogorov equation
\begin{equation*}
\frac{d}{d t}F_t=LF_t, \qquad F_t\bigr|_{t=0}=F_0.
\end{equation*}
However, the construction of a generic Markov process, either on $\Ga^2$ or 
$\Ga$, is essentially an open problem (for some particular cases on $\Ga$
see e.g.~\cite{GK2006, GK2008}).

In spite of this technical difficulty, in applications it turns out that
we need a knowledge on certain characteristics of the stochastic 
evolution in terms of mean values rather than pointwise. These 
characteristics concern e.g.~observables, that is, functions defined on 
$\Ga^2$, which expected values are given by
\begin{equation*}
\langle F,\mu\rangle:=\int_{\Ga^2}\,d\mu(\ga ^+,\ga ^-)F(\ga ^+,\ga
^-),
\end{equation*}
being $\mu$ a probability measure on $\Ga^2$, that is, a state of the system.
This leads to the following time evolution problem on states,
\begin{equation}\label{StateEvol}
\frac{d}{d t}\langle F,\mu_t\rangle=\langle LF,\mu_t\rangle, \qquad \mu_t\bigr|_{t=0}=\mu_0.
\end{equation}
For $F$ being of the type $F=\KK G$, $G\in\Bbs$, \eqref{StateEvol} may be 
rewritten in terms of the correlation functionals $k_t=k_{\mu_t}$ 
corresponding to the measures $\mu_t$, provided these functionals exist (or, 
more generally, in terms of correlation measures $\rho_t=\rho_{\mu_t}$), 
yielding
\begin{equation}\label{QBKE}
\frac{d}{dt}\langle\!\langle G,k_t\rangle\!\rangle=\langle\!\langle \hat{L}G,k_t\rangle\!\rangle, \qquad k_t\bigr|_{t=0}=k_0,
\end{equation}
where $\hat{L}:=\KK ^{-1} L \KK $ (cf.~Remark \ref{remark}) and 
$\langle\!\langle\cdot,\cdot\rangle\!\rangle$ is the usual pairing
\begin{equation}\label{duality}
\langle\!\langle
G,k\rangle\!\rangle:=\int_{\Ga^2_0}d\lambda^2(\eta^+,\eta^-)\,G(\eta^+,\eta^-)\,k(\eta^+,\eta^-).
\end{equation}

Of course, a strong version of equation \eqref{QBKE} is
\begin{equation}\label{CFE}
\frac{d}{dt}k_t=\hat{L}^*k_t, \qquad k_t\bigr|_{t=0}=k_0,
\end{equation}
for $\hat{L}^*$ being the dual operator of $\hat{L}$ in the sense defined in 
\eqref{duality}. One may associate to any function $k$ on $\Ga_0^2$ a 
double sequence $\bigl\{k^{(n,m)} \bigr\}_{n,m\in\N_0}$, where $k^{(n,m)}:=k\!\!\upharpoonright_{\{(\eta^+,\eta^-)\in\Ga_0^2:\vert\eta^+\vert=n, \vert\eta^-\vert =m\}}$ is a 
symmetric function on $(\X )^n\times(\X )^m$. This means that related to 
\eqref{CFE} one has a countable infinite number of equations having an 
hierarchical structure,
\begin{equation}
\frac {d}{d t}k_t^{(n,m)}=(\hat{L}^*k_t)^{(n,m)},\qquad k^{(n,m)}_t\bigr|_{t=0}=k^{(n,m)}_0\quad n,m\in\mathbb{N}_0,
\label{nova}
\end{equation}
where each equation only depends on a finite number of coordinates. As a 
result, we have reduced the infinite-dimensional problem \eqref{StateEvol}
to the infinite system of equations \eqref{nova}. However, it is convenient to 
recall here that, due to \eqref{QBKE}, we are only interesting in weak 
solutions to \eqref{nova}. 

Evolutions \eqref{QBKE}, \eqref{CFE} are obviously connected with an initial 
value problem on quasi-observables, that is, functions defined on $\Ga_0^2$, 
namely,
\begin{equation}\label{QOE}
\frac{d}{dt}G_t=\hat{L}G_t, \qquad G_t\bigr|_{t=0}=G_0.
\end{equation}
As explained before, one may also associate to \eqref{QOE} a double 
sequence, and thus, a countable infinite number of equations having also an 
hierarchical structure. In concrete cases, sometimes equation \eqref{QOE} 
appears easier to be analyzed in a suitable space. Having a 
solution to \eqref{QOE}, by duality \eqref{duality}, one might find a solution 
to \eqref{QBKE}. For instance, for birth-and-death systems on $\Ga$, this 
scheme has been accomplished in \cite{FKK11} through the derivation of 
semigroup evolutions for quasi-observables and correlation functions. Those 
results can be naturally extended to the multicomponent case. However, on each 
concrete application of other multicomponent models, namely, the conservative 
models considered below, the explicit form of the rates determines specific 
assumptions, and thus a specific analysis, which only hold for that concrete 
application.   

According to the considerations above, there is a close connection between the 
Markov evolution \eqref{StateEvol} and the hierarchical equations \eqref{CFE} 
and \eqref{QOE}. Of course, to derive solutions to \eqref{StateEvol} from 
solutions to \eqref{QBKE} an additional analysis is needed, namely, to 
distinguish the correlation functionals from the set of solutions to 
\eqref{QBKE}.

In what follows we derive explicit formulas for $\hat{L}, \hat{L}^*$ of 
general birth-and-death, hopping and flipping particle systems. For each case, 
explicit expressions are first derived on the space $\Bbs$, and then 
extended to linear operators on suitable Banach spaces.  
  
\section{Birth-and-death dynamics}\label{Section3}

\subsection{Hierarchical equations}

In a birth-and-death dynamics of a stochastic spatial type model, at
each random moment of time, particles randomly appear or disappear according 
to birth and death rates which depend on the configuration of the whole system 
at that time. As each particle is of one of the two possible types, 
$+$ and $-$, generators for such systems are informally described as the sum
of birth-and-death generators $L_{+}$ and $L_{-}$ of the $+$-system
and the $-$-system of particles involved. That is,
\begin{equation}
L=L_{+}+L_{-},\label{M9}
\end{equation}
where
\begin{align}
(L_{+}F)(\ga ^+,\ga ^-):=&\sum_{x\in\ga ^+}d^+(x,\ga ^+\setminus
x,\ga ^-)
\left(F(\ga ^+\setminus x,\ga ^-)-F(\ga ^+,\ga ^-)\right)\label{R1}\\
&+ \int_{\X }dx\, b^+(x,\ga ^+,\ga ^-)\left(F(\ga ^+\cup x,\ga
^-)-F(\ga ^+,\ga ^-)\right)\nonumber
\end{align}
and
\begin{align}
(L_{-}F)(\ga ^+,\ga ^-):=&\sum_{y\in\ga ^-}d^-(y,\ga ^+,\ga
^-\setminus y)
\left(F(\ga ^+,\ga ^-\setminus y)-F(\ga ^+,\ga ^-)\right)\label{Madeira}\\
&+\int_{\X }dy\, b^-(y,\ga ^+,\ga ^-)\left(F(\ga ^+,\ga ^-\cup
y)-F(\ga ^+,\ga ^-)\right).\notag
\end{align}
We observe that in \eqref{R1} the coefficient $d^+(x,\ga
^+,\ga ^-)\geq 0$ indicates the rate at which a + particle located
at $x\in\ga ^+$ dies or disappears, while $b^+(x,\ga ^+,\ga ^-)\geq
0$ indicates the rate at which, given a configuration $(\ga ^+,\ga
^-)$, a new + particle is born or appears at a site $x$. A similar
interpretation holds for the rates $d^-$ and $b^-$ appearing in
\eqref{Madeira}.

In order to give a meaning to \eqref{R1}, \eqref{Madeira}, in what follows we 
assume that $d^\pm,  b^\pm\geq 0$ are measurable functions such that, for 
a.a.~$x\in\X$, $d^\pm(x,\cdot,\cdot), b^\pm(x,\cdot,\cdot)$ are 
$\mathcal{B}(\Ga^2_0)$-measurable functions and, for 
$(\eta^+,\eta^-)\in\Ga^2_0$, $d^\pm(\cdot,\eta^+,\eta^-)$, 
$b^\pm(\cdot,\eta^+,\eta^-)\in L^1_{\mathrm{loc}}(\X,dx)$. These conditions are 
sufficient to ensure that for any $F\in\mathcal{FP}(\Ga^2)=\KK (\Bbs)$ the 
expression for $LF$, defined above, is well-defined at least on $\Ga^2_0$, 
which allows to define $\KK ^{-1}L\KK G$ (Remark \ref{remark}). This means, in 
particular, that for functions $G\in \Bbs$,
\begin{equation*}
(\hat{L}G)(\eta^+,\eta^-)=(\KK ^{-1}L\KK G)(\eta^+,\eta^-)
\end{equation*}
is well-defined on $\Ga^2_0$. In addition, the previous conditions allow to 
introduce the functions
\begin{align}\label{Ddef}
D^\pm(x,\xi^+,\xi^-,\eta^+,\eta^-)&:=\bigl(\KK ^{-1}d^\pm(x,\cdot\cup\xi^+,\cdot\cup\xi^-)\bigr)(\eta^+,\eta^-),\\
B^\pm(x,\xi^+,\xi^-,\eta^+,\eta^-)&:=\bigl(\KK
^{-1}b^\pm(x,\cdot\cup\xi^+,\cdot\cup\xi^-)\bigr)(\eta^+,\eta^-),\label{Bdef}
\end{align}
for a.a.~$x\in\X $, $(\eta^+,\eta^-), (\xi^+,\xi^-)\in\Ga^2_0$ such that 
$\eta^\pm\cap\xi^\pm=\emptyset$. We set
\begin{align*}
D^\pm_x(\eta^+,\eta^-)&:=D^\pm(x,\emptyset,\emptyset,\eta^+,\eta^-),\\
B^\pm_x(\eta^+,\eta^-)&:=B^\pm(x,\emptyset,\emptyset,\eta^+,\eta^-).
\end{align*}

\begin{proposition}\label{Prop1} The action of $\hat{L}$ on functions 
$G\in\Bbs $ is given for any $(\eta^+,\eta^-)\in\Ga_0^2$ by
\begin{align}\label{Lhatexpr}
&(\hat{L}G)(\eta^+,\eta^-)\\=&-\sum_{\substack{\xi^+\subset\eta^+\\
\xi^-\subset\eta^-}}G(\xi^+,\xi^-)\sum_{x\in\xi^+}D^+\bigl(x,\xi^+\setminus
x,\xi^-,\eta^+\setminus\xi^+,
\eta^-\setminus\xi^-\bigr)\nonumber\\&+\,\sum_{\substack{\xi^+\subset\eta^+\\
\xi^-\subset\eta^-}}\int_{\X }dx\,G(\xi^+\cup
x,\xi^-)B^+\bigl(x,\xi^+,\xi^-,\eta^+\setminus\xi^+,
\eta^-\setminus\xi^-\bigr)\nonumber\\&-\sum_{\substack{\xi^+\subset\eta^+\\
\xi^-\subset\eta^-}}G(\xi^+,\xi^-)\sum_{y\in\xi^-}D^-\bigl(y,\xi^+,\xi^-\setminus
y,\eta^+\setminus\xi^+,
\eta^-\setminus\xi^-\bigr)\nonumber\\&+\sum_{\substack{\xi^+\subset\eta^+\\
\xi^-\subset\eta^-}}\int_{\X }dy\,G(\xi^+,\xi^-\cup
y)B^-\bigl(y,\xi^+,\xi^-,\eta^+\setminus\xi^+,
\eta^-\setminus\xi^-\bigr).\nonumber
\end{align}
\end{proposition}

\begin{proof} We begin by observing that the integrability property of 
$b^\pm,d^\pm$ implies that $B^\pm,D^\pm$ are locally integrable on $\X$, and 
thus, for $G\in \Bbs$, both integrals appearing in \eqref{Lhatexpr} are finite.

Since $L$ is of the form \eqref{M9}, the proof of this result 
reduces to show the statement for $L_{+}$ and $L_{-}$. For this 
purpose, first we observe that from definition \eqref{M1} of the 
$\KK$-transform, for any $(\ga ^+,\ga ^-)\in\Ga_0^2$ we have
\begin{align*}
(\KK G)(\ga ^+\setminus x,\ga ^-)-(\KK G)(\ga ^+,\ga
^-)&=-\sum_{\eta^+\subset\ga ^+\setminus x}\sum_{\eta^-\subset\ga
^-}G(\eta^+\cup x,\eta^-),\\(\KK G)(\ga ^+\cup x,\ga ^-)-(\KK G)(\ga
^+,\ga ^-)&=\sum_{\eta^+\subset\ga ^+}\sum_{\eta^-\subset\ga
^-}G(\eta^+\cup x,\eta^-),\ x\notin\ga^+.
\end{align*}
We observe, in addition, that given a function $H$ of the form
\[
H(\ga ^+,\ga ^-):=\sum_{x\in\ga ^+}h(x,\ga ^+\setminus x,\ga ^-),
\]
for some suitable $h:\X\times\Ga^2\rightarrow \R$, it follows from definition
\eqref{inverseK2} of $\KK^{-1}$ that
\begin{equation}\label{rulesum}
(\KK ^{-1}H)(\eta^+,\eta^-)=\sum_{x\in\eta^+} (\KK^{-1}h)(x,\eta^+\setminus x,\eta^-).
\end{equation}
As a result, using definitions \eqref{Ddef}, \eqref{Bdef} of $B^+,D^+$ and 
the algebraic property \eqref{combin} of the $\Star$-convolution, we obtain 
the following expression for $\hat{L}_+G:=\KK^{-1}L_+\KK G$, $G\in\Bbs$,
\begin{align*}
 (\hat{L}_+G)(\eta^+,\eta^-)=&-\sum_{x\in\eta^+}\left(D_x^+\Star G(\cdot\cup
x,\cdot)\right)(\eta^+\setminus x,\eta^-)
\\&+\int_{\X }dx\,\left(B_x^+\Star G(\cdot\cup x,\cdot)\right)(\eta^+,\eta^-),
\end{align*}
which, by definition \eqref{dd} of the $\Star$-convolution, is equivalent to
\begin{align*}
&(\hat{L}_+G)(\eta^+,\eta^-)\\=&-\sum_{x\in\eta^+}\sum_{\substack{\xi^+\subset\eta^+\setminus
x\\ \xi^-\subset\eta^-}}G(\xi^+\cup x,\xi^-)
\sum_{\substack{\zeta^+\subset\xi^+\\ \zeta^-\subset\xi^-}}
D_x^+(((\eta^+\setminus
x)\setminus\xi^+)\cup\zeta^+,(\eta^-\setminus\xi^-)\cup\zeta^-)\\&+\int_{\X
}dx\,\sum_{\substack{\xi^+\subset\eta^+\\
\xi^-\subset\eta^-}}G(\xi^+\cup x,\xi^-)
\sum_{\substack{\zeta^+\subset\xi^+\\ \zeta^-\subset\xi^-}}
B_x^+((\eta^+\setminus\xi^+)\cup\zeta^+,(\eta^-\setminus\xi^-)\cup\zeta^-).
\end{align*}

Given a $\mathcal{B}(\Ga_0^2)$-measurable function $G'$ and
$(\eta_1^+,\eta_1^-)$, $(\eta_2^+,\eta_2^-)\in\Ga_0^2$, from the equality 
\[
(\KK G')(\eta_1^+\cup\eta_2^+,\eta_1^-\cup\eta_2^-)=\sum_{\xi_1^+\subset\eta_1^+}\sum_{\xi_2^+\subset\eta_2^+}\sum_{\xi_1^-\subset\eta_1^-}\sum_{\xi_2^-\subset\eta_2^-}G'(\xi_1^+\cup\xi_2^+,\xi_1^-\cup\xi_2^-)
\]
it follows that, for $F'(\eta^+,\eta^-):=(\KK G')(\eta^+,\eta^-)$, we have
\[
\bigl(\KK ^{-1} F'(\cdot\cup\xi^+,\cdot\cup\xi^-)\bigr)(\eta^+,\eta^-)=\bigl(\KK
G'(\eta^+\cup\cdot,\eta^-\cup\cdot)\bigr)(\xi^+,\xi^-).
\]
This applies, in particular, to $G'=D^+_x$, $F'=d^+(x,\cdot,\cdot)$ as well as
to $G'=B^+_x$, $F'=b^+(x,\cdot,\cdot)$, yielding
\begin{align*}
&(\hat{L}_+G)(\eta^+,\eta^-)\\=&-\sum_{x\in\eta^+}\sum_{\substack{\xi^+\subset\eta^+\setminus
x\\ \xi^-\subset\eta^-}}G(\xi^+\cup x,\xi^-) \bigl(\KK
^{-1}d^+(x,\cdot\cup\xi^+,\cdot\cup\xi^-)\bigr)((\eta^+\setminus
x)\setminus\xi^+, \eta^-\setminus\xi^-)\\&+\int_{\X
}dx\,\sum_{\substack{\xi^+\subset\eta^+\\
\xi^-\subset\eta^-}}G(\xi^+\cup x,\xi^-) \bigl(\KK
^{-1}b^+(x,\cdot\cup\xi^+,\cdot\cup\xi^-)\bigr)(\eta^+\setminus\xi^+,
\eta^-\setminus\xi^-).
\end{align*}
The required expression for $\hat{L}_+$ then follows by interchanging the two 
sums appearing in the first summand and using \eqref{Ddef}, \eqref{Bdef}. 
Similar arguments applied to $L_{-}$ complete the proof.
\end{proof}

As we have mentioned in Subsection \ref{Subsection22}, $\hat{L}^*$ is 
defined on any $\mathcal{B}(\Ga_0^2)$-measurable function $k$ with respect to 
which the following equality holds
\[
\int_{\Ga_0^2} d\la^2\,\hat{L} G \, k =\int_{\Ga_0^2} d\la^2\, G \, \hat{L}^*k
\]
for all $G\in\Bbs$. In the next subsection we will give a meaning to 
$\hat{L}^*$ as an operator defined on a proper space of functions on 
$\Ga_0^2$. Before that, we derive an explicit expression for $\hat{L}^*k$,
$k\in\Bbs$.

\begin{proposition}\label{prop5}
Assume that for all $\La\in\B _c(\X )$ and all $n,m\in\N_0$,
\begin{multline*}
A^+_{\La,m,n}:=\int_{\Ga^{(n,m)}_\La}d\la^2(\eta^+,\eta^-)\sum_{\substack{\xi^+\subset\eta^+\\\xi^-\subset\eta^-}}\Biggl(\,\sum_{x\in\xi^+}
\bigl| D^+(x,\xi^+\setminus x,\xi^-,\eta^+\setminus\xi^+,\eta^-\setminus\xi^-)\bigr|\\+\int_\La dx
\bigl| B^+(x,\xi^+,\xi^-,\eta^+\setminus\xi^+,\eta^-\setminus\xi^-)\bigr|\Biggr)<\infty
\end{multline*}
and
\begin{multline*}
A^-_{\La,m,n}:=\int_{\Ga^{(n,m)}_\La}d\la^2(\eta^+,\eta^-)\sum_{\substack{\xi^+\subset\eta^+\\\xi^-\subset\eta^-}}\Biggl(\,\sum_{y\in\xi^-}
\bigl| D^-(y,\xi^+,\xi^-\setminus y,\eta^+\setminus\xi^+,\eta^-\setminus\xi^-)\bigr|\\+\int_\La dy
\bigl| B^-(y,\xi^+,\xi^-,\eta^+\setminus\xi^+,\eta^-\setminus\xi^-)\bigr|\Biggr)<\infty,
\end{multline*}
where 
$\Ga_\La^{(n,m)}:=\left(\Ga_\La^{(n)}\times\Ga_\La^{(m)}\right)\cap\Ga_0^2$. 
Then, for each $k\in\Bbs$,
\begin{align}\label{hatLexprast}
&(\hat{L}^*k)(\eta^+,\eta^-)\\=&-\sum_{x\in\eta^+}\int_{\Ga_0^2}d\la^2(\xi^+,
\xi^-)k(\eta^+\cup\xi^+,\eta^-\cup\xi^-)D^+\bigl(x,\eta^+\setminus x,\eta^-,\xi^+,
\xi^-\bigr)\nonumber\\
&+\sum_{x\in\eta^+}\int_{\Ga_0^2}d\la^2(\xi^+,
\xi^-)k((\eta^+\setminus x)\cup\xi^+,\eta^-\cup\xi^-)B^+\bigl(x,\eta^+\setminus x,\eta^-,\xi^+,
\xi^-\bigr)\nonumber\\
&-\sum_{y\in\eta^-}\int_{\Ga_0^2}d\la^2(\xi^+,
\xi^-)k(\eta^+\cup\xi^+,\eta^-\cup\xi^-)D^-\bigl(y,\eta^+,\eta^-\setminus y,\xi^+,
\xi^-\bigr)\nonumber\\
&+\sum_{y\in\eta^-}\int_{\Ga_0^2}d\la^2(\xi^+,
\xi^-)k(\eta^+\cup\xi^+,(\eta^-\setminus y)\cup\xi^-)B^-\bigl(y,\eta^+,\eta^-\setminus y,\xi^+,
\xi^-\bigr),\nonumber
\end{align}
for $\lambda^2$-a.a.~$(\eta^+,\eta^-)\in\Ga_0^2$.
\end{proposition}

\begin{proof}
By the definition of the space $\Bbs$, given $G,k\in\Bbs$ there are 
$\La\in\B_c(\X )$, $N\in\N$, $C>0$ such that
\[
|G|,|k|\leq C\1_{\left(\bigsqcup_{n=0}^{N}\Gamma_{\Lambda}^{(n)}\times\bigsqcup_{n=0}^{N}\Gamma_{\Lambda}^{(n)}\right)\cap\Ga_0^2}, 
\]
where $\1_\cdot$ denotes the indicator function of a set. Therefore,
\begin{align*}
&\int_{\Ga_0^2}d\la^2(\eta^+,\eta^-)\sum_{\substack{\xi^+\subset\eta^+\\
\xi^-\subset\eta^-}}\Biggl(\,\bigl|G(\xi^+,\xi^-)\bigr|\sum_{x\in\xi^+}
\bigl|D^+\bigl(x,\xi^+\setminus x,\xi^-,\eta^+\setminus\xi^+,
\eta^-\setminus\xi^-\bigr)\bigr|\nonumber\\&\qquad\qquad+\int_{\X
}dx\,\bigl|G(\xi^+\cup
x,\xi^-)\bigr|\bigl|B^+\bigl(x,\xi^+,\xi^-,\eta^+\setminus\xi^+,
\eta^-\setminus\xi^-\bigr)\bigr|\Biggr)
\bigl|k(\eta^+,\eta^-)|\\\leq&\,C^2\sum_{m,n=0}^NA^+_{\La,m,n}<\infty.
\end{align*}
This shows that the product $(\hat{L}_+G)k$ is integrable over $\Gamma_0^2$ 
with respect to the measure $\lambda^2$. Moreover, using the expression for 
$\hat{L}_+ G$ (derive in Proposition \ref{Prop1} and its proof) and Lemma 
\ref{Lmm2} we obtain
\begin{align*}
&\int_{\Ga_0^2} d\la^2(\eta^+,\eta^-)(\hat{L}_+ G) (\eta^+,\eta^-)\,
k(\eta^+,\eta^-) \\=& -\int_{\Ga_0^2}
d\la^2(\eta^+,\eta^-)\int_{\Ga_0^2}d\la^2(\xi^+,\xi^-)
k(\eta^+\cup\xi^+,\eta^-\cup\xi^-)\\&\qquad\qquad\times G(\xi^+,\xi^-)\sum_{x\in\xi^+}D^+\bigl(x,\xi^+\setminus
x,\xi^-,\eta^+, \eta^-\bigr)\nonumber\\
&+\int_{\Ga_0^2}
d\la^2(\eta^+,\eta^-)\int_{\Ga_0^2}d\la^2(\xi^+,\xi^-)
k(\eta^+\cup\xi^+,\eta^-\cup\xi^-)\\&\qquad\qquad\times\int_{\X
}dx\,G(\xi^+\cup x,\xi^-)B^+\bigl(x,\xi^+,\xi^-,\eta^+,
\eta^-\bigr),
\end{align*}
where a second application of Lemma \ref{Lmm2} to the latter summand leads to 
the expression for $\hat{L}_+^*$. Similar considerations yield an
expression for $\hat{L}_-^*$.
\end{proof}

\subsection{Definition of operators}\label{Subsection31}

For each $C>0$, let us consider the Banach space
\begin{equation}
\L _C:=L^1(\Ga_0^2,\la_C^2)\label{Bielefeld}
\end{equation}
with the usual norm 
\[
\|G\|_{\L_C}:=\int_{\Ga_0^2}d\la^2(\eta^+,\eta^-)\,|G(\eta^+,\eta^-)|\,C^{|\eta^+|+|\eta^-|}.
\] 
Assume that there is a function $N:\Ga_0^2\rightarrow\R$ such that 
\begin{equation}\label{taquase}
\int_{\Ga_\La^{(n,m)}}d\lambda^2(\eta^+,\eta^-)\,N(\eta^+,\eta^-)<\infty\quad \mathrm{for\,\,all}\,\,n,m\in\N
\mathrm{\,\,and\,\,all\,\,} \Lambda \in \B _c(\X )
\end{equation}
and, for $\la^2$-a.a.~$(\eta^+,\eta^-)\in\Ga_0^2$,
\begin{align}
&\sum_{x\in\eta^+}\Bigl\|D^+\bigl(x,\eta^+\setminus
x,\eta^-,\cdot, \cdot\bigr)\Bigr\|_{\L
_C}+\frac{1}{C}\sum_{x\in\eta^+}\Bigl\|B^+\bigl(x,\eta^+\setminus
x,\eta^-,\cdot, \cdot\bigr)\Bigr\|_{\L
_C}\nonumber\\&+\sum_{y\in\eta^-}\Bigl\|D^-\bigl(y,\eta^+,\eta^-\setminus
y,\cdot, \cdot\bigr)\Bigr\|_{\L
_C}+\frac{1}{C}\sum_{y\in\eta^-}\Bigl\|B^-\bigl(y,\eta^+,\eta^-\setminus
y,\cdot, \cdot\bigr)\Bigr\|_{\L _C}\nonumber\\ \leq&\,N(\eta^+,\eta^-)<\infty.\label{bdd}
\end{align}
This allows to define the set
\[
\D :=\D_{N,C}:=\bigl\{ G\in\L _C\bigm| NG\in\L _C\bigr\}.
\]
It is clear that $\Bbs\subset\D $, which implies that also $\D$ is dense in 
$\L _C$.

\begin{proposition}
Assume that integrability conditions \eqref{taquase}, \eqref{bdd} hold. 
Then, equality \eqref{Lhatexpr} provides a densely defined linear operator 
$\hat{L}$ in $\L _C$ with domain $\D$. In particular, for any $G\in\D$, the 
right-hand side of \eqref{Lhatexpr} is $\la^2$-a.e.~well-defined on 
$\Gamma_0^2$.
\end{proposition}

\begin{proof}
Given a $G\in\D$, an application of Lemma \ref{Lmm2} to the expression 
corresponding to $\hat{L}_+$ (derived in Proposition \ref{Prop1} and its 
proof) yields
\begin{align*}
&\bigl\|\hat{L}_+G\bigr\|_{\L _C}\\\leq&\int_{\Ga_0^2}
d\la^2(\eta^+,\eta^-)\,C^{|\eta^+|+|\eta^-|}\int_{\Ga_0^2}d\la^2(\xi^+,\xi^-)\,C^{|\xi^+|+|\xi^-|}
\bigl|G(\xi^+,\xi^-)\bigr|\\&\qquad\qquad\times\sum_{x\in\xi^+}\bigl|D^+\bigl(x,\xi^+\setminus
x,\xi^-,\eta^+, \eta^-\bigr)\bigr|\\&+\,\int_{\Ga_0^2}
d\la^2(\eta^+,\eta^-)\,C^{|\eta^+|+|\eta^-|}\int_{\Ga_0^2}d\la^2(\xi^+,\xi^-)\,C^{|\xi^+|+|\xi^-|}\\&\qquad\qquad\times
\int_{\X }dx\,\bigl|G(\xi^+\cup
x,\xi^-)\bigr|\bigl|B^+\bigl(x,\xi^+,\xi^-,\eta^+,
\eta^-\bigr)\bigr|,
\end{align*}
and a similar estimate holds for $\|\hat{L}_-G\|_{\L _C}$. As a result,
\[
\bigl\|\hat{L}G\bigr\|_{\L _C}\leq\bigl\|NG\bigr\|_{\L _C}<\infty.\qedhere
\]
\end{proof}

Let us consider the dual space $( \L _{C})^{\prime }$, which can be realized by 
the Banach space 
\begin{equation*}
{\K }_{C}:=\left\{ k:\Ga^2_{0}\rightarrow {\R}\,\Bigm| k\cdot
C^{-|\cdot^+ |-|\cdot^-|}\in L^{\infty }(\Ga^2_{0},\la^2 )\right\}
\end{equation*} 
with the norm
\[
\Vert k\Vert _{{\K }_{C}}:=\Vert C^{-|\cdot^+ |-|\cdot^-|}k\Vert
_{L^{\infty }(\Ga^2_{0},\la^2 )}.
\]
The duality between the Banach spaces $\L_{C}$ and ${\K }_{C}$  is given by
\eqref{duality} with $\left\vert\left\langle \!\left\langle G,k\right\rangle \!\right\rangle
\right\vert \leq \Vert G\Vert _{\L _C}\cdot \Vert k\Vert _{{ \K
}_{C}}$. We observe that if $k\in\K _{C}$, then
\begin{equation}\label{normineq}
|k(\eta^+,\eta^-)|\leq \Vert k\Vert _{\K
_{C}}\,C^{|\eta^+|+|\eta^-|}
\end{equation}
for $\la^2$-a.a.~$(\eta^+,\eta^-) \in \Ga^2_{0}$.

\begin{proposition}\label{prop21}
Assume that integrability conditions \eqref{taquase}, \eqref{bdd} hold. In 
addition, assume that there are constants $A>0$, $M\in\N$, $\nu\geq1$ such that
\begin{equation}\label{notbig}
N(\eta^+,\eta^-)\leq A\bigl(1+|\eta^+|+|\eta^-|\bigr)^M\nu^{|\eta^+|+|\eta^-|}.
\end{equation}
Then, equality \eqref{hatLexprast} provides a linear operator $\hat{L}^*$ in 
$\K _C$ with domain $\K _{\alpha C}$, $\alpha\in\bigl(0,\frac{1}{\nu}\bigr)$. 
In particular, given a 
$k\in\K _{\alpha C}$ for some $\alpha\in\bigl(0,\frac{1}{\nu}\bigr)$, the 
right-hand side of \eqref{hatLexprast} is $\la^2$-a.e.~well-defined on 
$\Gamma_0^2$.
\end{proposition}

\begin{proof}
For some $\alpha\in\bigl(0,\frac{1}{\nu}\bigr)$, let $k\in\K _{\alpha C}$. 
Then, using the expression corresponding to $\hat{L}_+^*$, defined in 
Proposition \ref{prop5} and its proof, for 
$\la^2$-a.a.~$(\eta^+,\eta^-)\in\Ga_0^2$ we obtain
\begin{align*}
&C^{-|\eta^+|-|\eta^-|}\bigl|(\hat{L}_+^*k)(\eta^+,\eta^-)\bigr|\\\leq&\,\|k\|_{\K
_{\alpha C}}\alpha
^{|\eta^+|+|\eta^-|}\sum_{x\in\eta^+}\int_{\Ga_0^2}d\la^2(\xi^+,
\xi^-)(\alpha
C)^{|\xi^+|+|\xi^-|}\\&\qquad\qquad\times\bigl|D^+\bigl(x,\eta^+\setminus
x,\eta^-,\xi^+,
\xi^-\bigr)\bigr|\\
&+\|k\|_{\K _{\alpha C}}(\alpha C)^{-1}\alpha
^{|\eta^+|+|\eta^-|}\sum_{x\in\eta^+}\int_{\Ga_0^2}d\la^2(\xi^+,
\xi^-)(\alpha
C)^{|\xi^+|+|\xi^-|}\\&\qquad\qquad\times\bigl|B^+\bigl(x,\eta^+\setminus
x,\eta^-,\xi^+, \xi^-\bigr)\bigr|,
\end{align*}
where we have used inequality \eqref{normineq}. A similar estimate holds for 
$C^{-|\eta^+|-|\eta^-|}\cdot$ $\bigl|(\hat{L}_-^*k)(\eta^+,\eta^-)\bigr|$. Both 
estimates combined with \eqref{notbig} lead to
\begin{align*}
C^{-|\eta^+|-|\eta^-|}\bigl|(\hat{L}^*k)(\eta^+,\eta^-)\bigr|
&\leq\frac{\|k\|_{\K _{\alpha C}}}{\alpha} \alpha
^{|\eta^+|+|\eta^-|}N(\eta^+,\eta^-)\\&\leq\frac{A\|k\|_{\K _{\alpha
C}}}{\alpha} (\alpha
\nu)^{|\eta^+|+|\eta^-|}\bigl(1+|\eta^+|+|\eta^-|\bigr)^M.
\end{align*}
Since $\alpha <1$, and thus $\alpha \nu <1$, an application of inequality
\begin{equation*}
(1+t)^{b}a^{t}\leq
\frac{1}{a}\left( \frac{b}{-e\ln a}\right) ^{b},\qquad b\geq 1,~a\in
\left( 0,1\right) ,~t\geq 0,
\end{equation*}
yields
\[
\bigl\|\hat{L}^*k\bigr\|_{\K _C}\leq\frac{A\|k\|_{\K _{\alpha
C}}}{\alpha}\frac{1}{\alpha\nu}\Bigl(\frac{M}{-e\ln(\alpha
\nu)}\Bigr)^M<\infty,
\]
completing the proof.
\end{proof}

\begin{remark}
Since the space $\L _C$ is not reflexive, a priori we cannot expect that the 
domain of $\hat{L}^*$ is dense in $\K _C$.
\end{remark}

\section{Conservative dynamics}\label{Section4}

In contrast to the birth-and-death dynamics, in the following dynamics
there is conservation on the total number of particles involved.

\subsection{Hopping particles: hierarchical equations}\label{subsection41}

Dynamically, in a hopping particle system, at each random moment of time 
particles randomly hop from one site to another according to a rate depending 
on the configuration of the whole system at that time. Since the particles are 
of two types, two situations may occur. The $\pm$ particles located in 
$\ga ^\pm$ hop over $\ga ^\pm$, or hop to sites in $\ga ^\mp$, thus changing 
its mark. In terms of generators these two different behaviors are informally 
described by
\begin{align*}
&(L_1F)(\ga^+,\ga^-)\\
:=&\sum_{x\in\ga ^+}\int_{\X }dx'\, c^+_1(x,x',\ga ^+\setminus x,\ga ^-)
\left(F(\ga ^+\setminus x\cup x',\ga ^-)-F(\ga ^+,\ga ^-)\right)\nonumber\\
&+ \sum_{y\in\ga ^-}\int_{\X }dy'\, c^-_1(y,y',\ga ^+,\ga ^-\setminus y)
\left(F(\ga ^+,\ga ^-\setminus y\cup y')-F(\ga ^+,\ga
^-)\right)\nonumber
\end{align*}
and
\begin{align}
&\left( L_2F\right) (\ga^+,\ga^-)\label{Y2}\\
:=&\sum_{x\in \ga ^+}\int_{\X }dy\,c^+_2\left(x, y, \ga ^+\setminus x,\ga
^-\right)
\left(F\left(\ga ^+\setminus x,\ga ^-\cup y\right)-F\left(\ga ^+,\ga ^-\right) \right)\nonumber\\
&+\sum_{y\in \ga ^{-}}\int_{\X }dx\,c^-_2\left(x,y,\ga ^+,\ga
^-\setminus y\right) \left(F\left(\ga ^+\cup x,\ga ^-\setminus y\right)
-F\left( \ga  ^{+},\ga  ^{-}\right) \right),\nonumber
\end{align}
respectively. Here the coefficient $c^+_1(x,x',\ga ^+,\ga ^-)\geq 0$
indicates the rate at which a + particle located at a position $x$
in a configuration $\ga ^+$ hops to a free site $x'$ keeping its mark, and 
$c^+_2(x,y,\ga ^+,\ga ^-)\geq 0$ indicates the rate at
which, given a configuration $(\ga ^+,\ga ^-)$, a + particle located
at a site $x\in\ga ^+$ hops to a free site $y$ and changes its mark to
$-$. A similar interpretation holds for the rates $c_i^-\geq 0$, $i=1,2$.

In what follows we assume that $c_i^\pm$, $i=1,2$, are measurable
functions such that, for a.a.~$x,y$, $c_i^\pm(x,y,\cdot,\cdot)$ are 
$\mathcal{B}(\Ga^2_0)$-measurable functions and, for 
$(\eta^+,\eta^-)\in\Ga_0^2$, $c_i^\pm(\cdot,\cdot,\eta^+,\eta^-)\in L^1_\mathrm{loc}(\X\times\X,dx\otimes dy)$.
Under these conditions, for each $F\in\mathcal{FP}(\Ga^2)=\KK(\Bbs)$, the 
expression for $L_iF$, $i=1,2$, is well-defined at least on $\Ga^2_0$, 
ensuring that for any $G\in \Bbs$
\begin{equation*}
\hat{L}_iG =\KK ^{-1}L_i\KK G
\end{equation*}
is well-defined on $\Ga^2_0$ (Remark \ref{remark}). Moreover, the above 
conditions allow to define the functions
\begin{equation*}
C_i^\pm(x,y,\xi^+,\xi^-,\eta^+,\eta^-):=\bigl(\KK
^{-1}c_i^\pm(x,y,\cdot\cup\xi^+,\cdot\cup\xi^-)\bigr)(\eta^+,\eta^-),\quad i=1,2,
\end{equation*}
for a.a.~$x,y\in\X $, $(\eta^+,\eta^-), (\xi^+,\xi^-)\in\Ga^2_0$ such that 
$\eta^\pm\cap\xi^\pm=\emptyset$. We set
\begin{equation*}
C^\pm_{i,x,y}(\eta^+,\eta^-):=C_i^\pm(x,y,\emptyset,\emptyset,\eta^+,\eta^-),
\quad i=1,2.
\end{equation*}

\begin{proposition}\label{Prop2} The action of $\hat{L}_i$, $i=1,2$, on 
functions $G\in\Bbs $ is given for any $(\eta^+,\eta^-)\in\Ga_0^2$ by
\begin{align}\label{L1+expr}
(\hat{L}_1 G)(\eta^+,\eta^-)=&\sum_{\substack{\xi^+\subset\eta^+ \\
\xi^-\subset\eta^-}}\sum_{x\in\xi^+}\int_{\X }dx'\,
\bigl(G(\xi^+\cup x'\setminus
x,\xi^-)-G(\xi^+,\xi^-)\bigr)\\&\qquad\qquad\times
C_1^+\bigl(x,x',\xi^+\setminus
x,\xi^-,\eta^+\setminus\xi^+,\eta^-\setminus
\xi^-\bigr)\nonumber\\&+\sum_{\substack{\xi^+\subset\eta^+ \\
\xi^-\subset\eta^-}}\sum_{y\in\xi^-}\int_{\X }dy'\,
\bigl(G(\xi^+,\xi^-\cup y'\setminus
y)-G(\xi^+,\xi^-)\bigr)\nonumber\\&\qquad\qquad\times
C_1^-\bigl(y,y',\xi^+,\xi^-\setminus
y,\eta^+\setminus\xi^+,\eta^-\setminus \xi^-\bigr),\nonumber
\end{align}
and
\begin{align}
(\hat{L}_2 G)(\eta^+,\eta^-)=&\sum_{\substack{\xi^+\subset\eta^+ \\
\xi^-\subset\eta^-}}\sum_{x\in\xi^+}\int_{\X
}dy\bigl(G(\xi^+\setminus x,\xi^-\cup
y)-G(\xi^+,\xi^-)\bigr)\label{L2+expr}\allowdisplaybreaks[0]\\&\qquad\qquad\times
C_2^+\bigl(x,y,\xi^+\setminus x,\xi^-,\eta^+\setminus
\xi^+,\eta^-\setminus
\xi^-\bigr)\nonumber\\&+\sum_{\substack{\xi^+\subset\eta^+ \\
\xi^-\subset\eta^-}}\sum_{y\in\xi^-}\int_{\X }dx\bigl(G(\xi^+\cup
x,\xi^-\setminus
y)-G(\xi^+,\xi^-)\bigr)\nonumber\\&\qquad\qquad\times
C_2^-\bigl(x,y,\xi^+,\xi^-\setminus y,\eta^+\setminus
\xi^+,\eta^-\setminus \xi^-\bigr).\nonumber
\end{align}
\end{proposition}

\begin{proof} We begin by observing that, similarly to the proof of 
Proposition \ref{Prop1}, the integrability property of $c_i^\pm$, $i=1,2$, on 
$\X$ is sufficient to ensure that, for any $G\in \Bbs$, all integrals 
appearing in \eqref{L1+expr}, \eqref{L2+expr} are finite. 

Since each $L_i$, $i=1,2$, is of the form $L_i= L_i^++L_i^-$, with $L_i^+$ 
concerning the $+$-system and $L_i^-$ the $-$-system, the proof reduces to 
prove the statement for each summand $L_i^+$, $L_i^-$, $i=1,2$. We will do it 
for $L_i^+$, $i=1,2$, being the proof for $L_i^-$, $i=1,2$, similar. For this 
purpose, first we observe that from definition \eqref{M1} of the 
$\KK$-transform, for any $(\ga ^+,\ga ^-)\in\Ga_0^2$ one has
\begin{align*}
(\KK G)(\ga ^+\setminus x\cup x',\ga ^-)&-(\KK G)(\ga ^+,\ga
^-)\\=&\,\bigl(\KK G(\cdot\cup x',\cdot)\bigr)(\gamma^+\setminus x,\gamma^-)-\bigl(\KK G(\cdot\cup x,\cdot)\bigr)(\gamma^+\setminus x,\gamma^-),\\(\KK G)(\ga ^+\setminus x,\ga ^-\cup y)&-(\KK G)(\ga
^+,\ga ^-)\\=&\,\bigl(\KK G(\cdot,\cdot\cup y)\bigr)(\gamma^+\setminus x,\gamma^-)-\bigl(\KK G(\cdot\cup x,\cdot)\bigr)(\gamma^+\setminus x,\gamma^-).
\end{align*}
This leads to
\begin{align*}
(\hat{L}_1^+G)(\eta^+,\eta^-)=&\sum_{x\in\eta^+}\int_{\X
}dx'\,\left(C^+_{1,x,x'}\Star \left(G(\cdot\cup
x',\cdot)-G(\cdot\cup x,\cdot)\right)\right)
(\eta^+\!\setminus\!x,\eta^-),\\
(\hat{L}_2^+G)(\eta^+,\eta^-)=&\sum_{x\in\eta^+}\int_{\X
}dy\,\left(C^+_{2,x,y}\Star \left(G(\cdot,\cdot\cup y)-G(\cdot\cup
x,\cdot)\right)\right) (\eta^+\!\setminus\!x,\eta^-),
\end{align*}
where we have used equality \eqref{rulesum}. Similar arguments used to prove 
Proposition \ref{Prop1} complete the proof for $L_i^+$, $i=1,2$.
\end{proof}

Concerning $\hat{L}_i^*$, $i=1,2$, one has the following explicit expressions.

\begin{proposition} Assume that for all $\La\in\B _c(\X )$ and all 
$n,m\in\N_0$,\begin{multline}
C_{1,\La,m,n}:=\int_{\Ga^{(n,m)}_\La}d\la^2(\eta^+,\eta^-)\int_\La
dx'\\ \times
\sum_{\substack{\xi^+\subset\eta^+\\\xi^-\subset\eta^-}}
\Biggl(\sum_{x\in\xi^+}\bigl|C_1^+\bigl(x,x',\xi^+\setminus x,\xi^-,
\eta^+\setminus\xi^+,\eta^-\setminus \xi^-\bigr)\bigr|\\
+ \sum_{y\in\xi^-}\bigl|C_1^-\bigl(y,x',\xi^+,\xi^-\setminus
y,\eta^+\setminus\xi^+, \eta^-\setminus
\xi^-\bigr)\bigr|\Biggr)<\infty\label{cond222}
\end{multline}
and
\begin{multline}
C_{2,\La,m,n}:=\int_{\Ga^{(n,m)}_\La}d\la^2(\eta^+,\eta^-)\int_\La
dx'\\ \times
\sum_{\substack{\xi^+\subset\eta^+\\\xi^-\subset\eta^-}}
\Biggl(\sum_{x\in\xi^+}\bigl|C_2^+\bigl(x,x',\xi^+\setminus x,\xi^-,
\eta^+\setminus\xi^+,\eta^-\setminus \xi^-\bigr)\bigr|\\
+ \sum_{y\in\xi^-}\bigl|C_2^-\bigl(x',y,\xi^+,\xi^-\setminus
y,\eta^+\setminus\xi^+, \eta^-\setminus
\xi^-\bigr)\bigr|\Biggr)<\infty,\label{cond2}
\end{multline}
where, as before, 
$\Ga_\La^{(n,m)}=\bigl(\Ga_\La^{(n)}\times\Ga_\La^{(m)}\bigr)\cap\Ga_0^2$. Then, 
for each $k\in\Bbs$,
\begin{align}\label{L1ast}
&(\hat{L}_1^*k)(\eta^+,\eta^-)\\=&
\nonumber\sum_{x\in\eta^+}\int_{\Ga_0^2} d\la^2(\xi^+,\xi^-)
\int_{\X }dx'\, k(\xi^+\cup\eta^+\cup x'\setminus
x,\xi^-\cup\eta^-)\\&\qquad\qquad\times
C_1^+\bigl(x',x,\eta^+\setminus
x,\eta^-,\xi^+,\xi^-\bigr)\nonumber\\&-
\sum_{x\in\eta^+}\int_{\Ga_0^2} d\la^2(\xi^+,\xi^-) \,
k(\xi^+\cup\eta^+,\xi^-\cup\eta^-)\nonumber\\&\qquad\qquad\times
\int_{\X }dx'C_1^+\bigl(x,x',\eta^+\setminus
x,\eta^-,\xi^+,\xi^-\bigr)\nonumber\\&+
\sum_{y\in\eta^-}\int_{\Ga_0^2} d\la^2(\xi^+,\xi^-) \int_{\X }dy'\,
k(\xi^+\cup\eta^+,\xi^-\cup\eta^-\cup y'\setminus
y)\nonumber\\&\qquad\qquad\times
C_1^-\bigl(y',y,\eta^+,\eta^-\setminus
y,\xi^+,\xi^-\bigr)\nonumber\\&-  \sum_{y\in\eta^-}\int_{\Ga_0^2}
d\la^2(\xi^+,\xi^-)
 k(\xi^+\cup\eta^+,\xi^-\cup\eta^-)\nonumber\\&\qquad\qquad\times \int_{\X }dy'\,C_1^-\bigl(y,y',\eta^+,\eta^-\setminus y,\xi^+,\xi^-\bigr),\nonumber
\end{align}
and
\begin{align}
&(\hat{L}_2^*k)(\eta^+,\eta^-)\label{L2ast}\\=&
\sum_{y\in\eta^-}\int_{\Ga_0^2} d\la^2(\xi^+,\xi^-) \int_{\X }dx\,
k(\xi^+\cup\eta^+\cup x,\xi^-\cup\eta^-\setminus
y)\nonumber\\&\qquad\qquad\times
C_2^+\bigl(x,y,\eta^+,\eta^-\setminus
y,\xi^+,\xi^-\bigr)\nonumber\\&-  \sum_{x\in\eta^+}\int_{\Ga_0^2}
d\la^2(\xi^+,\xi^-)
k(\xi^+\cup\eta^+,\xi^-\cup\eta^-)\nonumber\\&\qquad\qquad\times
\int_{\X }dy\, C_2^+\bigl(x,y,\eta^+\setminus
x,\eta^-,\xi^+,\xi^-\bigr)\nonumber\\&+
\sum_{x\in\eta^+}\int_{\Ga_0^2} d\la^2(\xi^+,\xi^-) \int_{\X
}dy\,k(\xi^+\cup\eta^+\setminus x,\xi^-\cup\eta^-\cup
y)\nonumber\\&\qquad\qquad\times C_2^-\bigl(x,y,\eta^+\setminus
x,\eta^-,\xi^+,\xi^-\bigr)\nonumber\\&-
\sum_{y\in\eta^-}\int_{\Ga_0^2} d\la^2(\xi^+,\xi^-) \int_{\X }dx\,
k(\xi^+\cup\eta^+,\xi^-\cup\eta^-)\nonumber\\&\qquad\qquad\times
C_2^-\bigl(x,y,\eta^+,\eta^-\setminus y,\xi^+,\xi^-\bigr),\nonumber
\end{align}
for $\lambda^2$-a.a.~$(\eta^+,\eta^-)\in\Ga_0^2$.
\end{proposition}

\begin{proof}
Similarly to the proof of Proposition \ref{prop5}, conditions \eqref{cond222},
\eqref{cond2} ensure that for any $G,k\in\Bbs$, one has 
$(\hat{L}_i^\pm G)k\in L^1(\Ga_0^2,\lambda^2)$, $i=1,2$. Moreover, for 
$\hat{L}_1^+$, the use of its expression, derived in Proposition \ref{Prop2} 
and its proof, leads through an application of Lemma \ref{Lmm2} to
\begin{align*}
&\int_{\Ga_0^2} d\la^2(\eta^+,\eta^-)\,(\hat{L}_1^+ G)
(\eta^+,\eta^-)\, k(\eta^+,\eta^-) \\=& \int_{\Ga_0^2}
d\la^2(\eta^+,\eta^-)\int_{\Ga_0^2}d\la^2(\xi^+,\xi^-)\,
k(\eta^+\cup\xi^+,\eta^-\cup\xi^-)\sum_{x\in\xi^+}\int_{\X
}dx'\\&\qquad\qquad\times \bigl(G(\xi^+\cup x'\setminus
x,\xi^-)-G(\xi^+,\xi^-)\bigr) C_1^+\bigl(x,x',\xi^+\setminus
x,\xi^-,\eta^+,\eta^-\bigr)\\=&\int_{\Ga_0^2}d\la^2(\xi^+,\xi^-)\,G(\xi^+,\xi^-)
\int_{\Ga_0^2} d\la^2(\eta^+,\eta^-) \\&\qquad\qquad\times
 \sum_{x'\in\xi^+}\int_{\X }dx\,k(\eta^+\cup\xi^+\cup x\setminus x',\eta^-\cup\xi^-)C_1^+\bigl(x,x',\xi^+\setminus x',\xi^-,\eta^+,\eta^-\bigr)\\& -\int_{\Ga_0^2}d\la^2(\xi^+,\xi^-)\,
G(\xi^+,\xi^-)\sum_{x\in\xi^+}\int_{\Ga_0^2} d\la^2(\eta^+,\eta^-)\,
k(\eta^+\cup\xi^+,\eta^-\cup\xi^-)\\&\qquad\qquad\times \int_{\X
}dx'\,C_1^+\bigl(x,x',\xi^+\setminus
x,\xi^-,\eta^+,\eta^-\bigr).
\end{align*}
Similarly, for $\hat{L}_2^+$, we obtain
\begin{align*}
&\int_{\Ga_0^2} d\la^2(\eta^+,\eta^-)\,(\hat{L}_2^+ G)
(\eta^+,\eta^-)\, k(\eta^+,\eta^-) \\=&\int_{\Ga_0^2}
d\la^2(\eta^+,\eta^-)\int_{\Ga_0^2}d\la^2(\xi^+,\xi^-)\,
k(\eta^+\cup\xi^+,\eta^-\cup\xi^-)\\&\qquad\qquad\times\sum_{x\in\xi^+}\int_{\X
}dy\,\bigl(G(\xi^+\setminus x,\xi^-\cup
y)-G(\xi^+,\xi^-)\bigr)\\&\qquad\qquad\qquad\qquad\times
C_2^+\bigl(x,y,\xi^+\setminus x,\xi^-,\eta^+,\eta^-\bigr).
\end{align*}
The rest of the proof follows now straightforwardly.
\end{proof}

\subsection{Hopping particles: definition of operators}

Assume that for each $i=1,2$ there is a function $N_i:\Ga_0^2\rightarrow\R$ 
such that 
\begin{equation}\label{taquase2}
\int_{\Ga_\La^{(n,m)}}d\lambda^2(\eta^+,\eta^-)\,N_i(\eta^+,\eta^-)<\infty\quad \mathrm{for\,\,all}\,\,n,m\in\N
\mathrm{\,\,and\,\,all\,\,} \Lambda \in \B _c(\X )
\end{equation}
and, for $\la^2$-a.a.~$(\eta^+,\eta^-)\in\Ga_0^2$,
\begin{align}
&\sum_{x\in\eta^+}\biggl(\Bigl\|\int_{\X
}dy\,C_1^+\bigl(x,y,\eta^+\setminus x,\eta^-,\cdot,
\cdot\bigr)\Bigr\|_{\L _C}\nonumber\\&\qquad\qquad\qquad+\Bigl\|\int_{\X
}dy\,C_1^+\bigl(y,x,\eta^+\setminus x,\eta^-,\cdot,
\cdot\bigr)\Bigr\|_{\L
_C}\biggr)\nonumber\\&+\sum_{y\in\eta^-}\biggl(\Bigl\|\int_{\X
}dx\,C_1^-\bigl(x,y,\eta^+,\eta^-\setminus y,\cdot,
\cdot\bigr)\Bigr\|_{\L _C}\nonumber\\&\qquad\qquad\qquad+\Bigl\|\int_{\X
}dx\,C_1^-\bigl(y,x,\eta^+,\eta^-\setminus y,\cdot,
\cdot\bigr)\Bigr\|_{\L _C}\biggr)\nonumber\\\leq &\,N_{1}(\eta^+,\eta^-)<\,\infty,\label{bdd1}
\end{align}
and
\begin{align}
&\sum_{x\in\eta^+}\biggl(\Bigl\|\int_{\X
}dy\,C_2^+\bigl(x,y,\eta^+\setminus x,\eta^-,\cdot,
\cdot\bigr)\Bigr\|_{\L _C}\nonumber\\&\qquad\qquad\qquad+\Bigl\|\int_{\X
}dy\,C_2^-\bigl(x,y,\eta^+\setminus x,\eta^-,\cdot,
\cdot\bigr)\Bigr\|_{\L
_C}\biggr)\nonumber\\&+\sum_{y\in\eta^-}\biggl(\Bigl\|\int_{\X
}dx\,C_2^+\bigl(x,y,\eta^+,\eta^-\setminus y,\cdot,
\cdot\bigr)\Bigr\|_{\L _C}\nonumber\\&\qquad\qquad\qquad+\Bigl\|\int_{\X
}dx\,C_2^-\bigl(x,y,\eta^+,\eta^-\setminus y,\cdot,
\cdot\bigr)\Bigr\|_{\L _C}\biggr)\nonumber\\\leq &\,N_2(\eta^+,\eta^-)<\,\infty.\label{bdd2}
\end{align}

Under these conditions, let us consider the sets
\[
\D _i:=\D_i (N_i,C):=\bigl\{ G\in\L _C\bigm| N_iG\in\L _C\bigr\}, \quad i=1,2,
\]
where $\L _C$ is the Banach space defined in \eqref{Bielefeld}. Of course, 
$\Bbs\subset\D _1\cap\D _2$, which implies that both $\D _1$ and $\D _2$ are 
dense in $\L _C$.

\begin{proposition}
Assume that integrability conditions \eqref{taquase2}, \eqref{bdd1}, 
\eqref{bdd2} hold. Then, equality \eqref{L1+expr} (resp., \eqref{L2+expr}) 
provides a densely defined linear operator $\hat{L}_1$ (resp., $\hat{L}_2$) in 
$\L _C$ with domain $\D _1$ (resp., $\D _2$). In particular, for any 
$G\in\D _1$ (resp., $G\in\D _2$), the right-hand side of \eqref{L1+expr} 
(resp., \eqref{L2+expr}) is $\la^2$-a.e.~well-defined on $\Gamma_0^2$.
\end{proposition}
\begin{proof}
We just estimate $\|\hat{L}_1^+G\|_{\L _C}$, being similar the 
estimate for $\hat{L}_1^-$. Given a $G\in\D _1$, an application 
of Lemma \ref{Lmm2} to the expression corresponding to $\hat{L}_1^+$ (derived 
in Proposition \ref{Prop2} and its proof) yields
\begin{align*}\|\hat{L}_1^+G\|_{\L _C}\leq&\int_{\Ga_0^2}d\la^2(\xi^+,\xi^-)
C^{|\xi^+|+|\xi^-|}\\&\qquad\qquad\times \sum_{x\in\xi^+}\int_{\X
}dx'\, \bigl(|G(\xi^+\cup x'\setminus
x,\xi^-)|+|G(\xi^+,\xi^-)|\bigr)\\&\qquad\qquad\qquad\times
\bigl\|C_1^+\bigl(x,x',\xi^+\setminus
x,\xi^-,\cdot,\cdot\bigr)\bigr\|_{\L _C}.
\end{align*}
This shows that $\|\hat{L}_1G\|_{\L _C}\leq\|\hat{L}_1^+G\|_{\L
_C}+\|\hat{L}_1^-G\|_{\L _C}\leq\|N_1G\|_{\L _C}<\infty$. The proof for 
$\hat{L}_2$ is analogous.
\end{proof}

Similar arguments used to prove Proposition \ref{prop21} lead to the next 
result.

\begin{proposition}
Assume that integrability conditions \eqref{taquase2}, \eqref{bdd1}, 
\eqref{bdd2} hold. In addition, assume that there are constants $A>0$, 
$M\in\N$, $\nu\geq1$ such that
\begin{equation*}
N_{i}(\eta^+,\eta^-)\leq A\bigl(1+|\eta^+|+|\eta^-|\bigr)^M\nu^{|\eta^+|+|\eta^-|},
\quad i=1,2.
\end{equation*}
Then, equality \eqref{L1ast} (resp., \eqref{L2ast}) provides a linear operator 
$\hat{L}_1^*$ (resp., $\hat{L}_2^*$) in $\K _C$ with domain $\K _{\alpha C}$, 
$\alpha\in\bigl(0,\frac{1}{\nu}\bigr)$. In particular, given a 
$k\in\K _{\alpha C}$ for some $\alpha\in\bigl(0,\frac{1}{\nu}\bigr)$, 
the right-hand side of \eqref{L1ast} (resp., \eqref{L2ast}) is 
$\la^2$-a.e.~well-defined on $\Gamma_0^2$.
\end{proposition}

\subsection{Flipping particles}

Dynamically, in a flipping particle system, at each random moment of
time particles randomly flip marks keeping their sites. In terms of generators 
this behavior is informally described by
\begin{align}\label{flip}
(L_0F)(\ga^+,\ga^-)=&\sum_{x\in\ga^+}a^+(x,\ga^+\setminus x,\ga^-)\bigl(F(\ga^+\setminus
x,\ga^-\cup x)-F(\ga^+,\ga^-)\bigr)\\&+\sum_{y\in\ga^-}a^-(x,\ga^+,\ga^-\setminus y)\bigl(F(\ga^+\cup
y,\ga^-\setminus y)-F(\ga^+,\ga^-)\bigr),\nonumber
\end{align}
where $a^+(x,\ga ^+,\ga ^-)\geq 0$ indicates the rate at which a
$+$-particle located at $x\in\ga ^+$ flips the mark to ``$-$''.
A similar interpretation holds for the rate $a^-\geq 0$ appearing in
\eqref{flip}. We observe that, formally, $L_0$ is a particular case of the 
mapping $L_2$ defined in \eqref{Y2} with
\[
c^\pm_2(x,y,\ga^+,\ga^-)=\delta(x-y)a^\pm(x,\ga^+,\ga^-).
\]
Therefore, the results obtained therein justify the results for $L_0$.
The proof of Proposition \ref{Prop11} below is then fully similar.

In what follows we assume that $a^\pm$ are measurable functions such that, for
a.a.~$x\in\X$, $a^\pm(x,\cdot,\cdot)$ are $\mathcal{B}(\Ga^2_0)$-measurable 
functions and, for $(\eta^+,\eta^-)\in\Ga^2_0$,
$a^\pm(\cdot,\eta^+,\eta^-)\in L^1_{\mathrm{loc}}(\R^d, dx)$. We set
\begin{equation*}
A^\pm(x,\xi^+,\xi^-,\eta^+,\eta^-):=\bigl(\KK
^{-1}a^\pm(x,\cdot\cup\xi^+,\cdot\cup\xi^-)\bigr)(\eta^+,\eta^-),
\end{equation*}
for a.a.~$x\in\R^d$ and $(\eta^+,\eta^-), (\xi^+,\xi^-)\in\Ga^2_0$ such that
$\eta^\pm\cap\xi^\pm=\emptyset$.

\begin{proposition}\label{Prop11} If $G\in
B_{\mathrm{bs}}(\Ga_0^2)$, then for any $(\eta^+,\eta^-)\in\Ga_0^2$
\begin{align*}
&(\hat{L}_0G)(\eta^+,\eta^-)\\=&\sum_{\substack{\xi^+\subset\eta^+ \\
\xi^-\subset\eta^-}}\sum_{x\in\xi^+}\bigl(G(\xi^+\setminus x,\xi^-\cup
x)-G(\xi^+,\xi^-)\bigr)\nonumber
A^+\bigl(x,\xi^+\setminus x,\xi^-,\eta^+\setminus
\xi^+,\eta^-\setminus
\xi^-\bigr)\nonumber\\&+\sum_{\substack{\xi^+\subset\eta^+ \\
\xi^-\subset\eta^-}}\sum_{y\in\xi^-}\bigl(G(\xi^+\cup
y,\xi^-\setminus
y)-G(\xi^+,\xi^-)\bigr)\nonumber
A^-\bigl(y,\xi^+,\xi^-\setminus y,\eta^+\setminus
\xi^+,\eta^-\setminus \xi^-\bigr).\nonumber
\end{align*}
If, in addition, there is a function $N_0:\Ga_0^2\rightarrow\R$ such that 
\[
\int_{\Ga_\La^{(n,m)}}d\lambda^2(\eta^+,\eta^-)\,N_0(\eta^+,\eta^-)<\infty\quad \mathrm{for\,\,all}\,\,n,m\in\N
\mathrm{\,\,and\,\,all\,\,} \Lambda \in \B _c(\X )
\]
and, for $\la^2$-a.a.~$(\eta^+,\eta^-)\in\Ga_0^2$,
\begin{align*}
&\sum_{x\in\eta^+}\Bigl(\bigl\|A^+\bigl(x,\eta^+\setminus x,\eta^-,\cdot,
\cdot\bigr)\bigr\|_{\L _C}+\bigl\|A^-\bigl(x,\eta^+\setminus x,\eta^-,\cdot,
\cdot\bigr)\bigr\|_{\L
_C}\Bigr)\nonumber\\&+\sum_{y\in\eta^-}\Bigl(\bigl\|A^+\bigl(y,\eta^+,\eta^-\setminus y,\cdot,
\cdot\bigr)\bigr\|_{\L _C}+\bigl\|A^-\bigl(y,\eta^+,\eta^-\setminus y,\cdot,
\cdot\bigr)\bigr\|_{\L _C}\biggr)\nonumber\\\leq &\, N_0(\eta^+,\eta^-)<\infty,
\end{align*}
then, for each $G\in\L_C$ such that $N_0G\in\L_C$, we have $\hat{L}_0G\in\L_C$.
Moreover, if there are $A>0$, $M\in\N$, $\nu\geq1$ such that
\begin{equation*}
N_{0}(\eta^+,\eta^-)\leq A\bigl(1+|\eta^+|+|\eta^-|\bigr)^M\nu^{|\eta^+|+|\eta^-|},
\end{equation*}
then, $\hat{L}_0^*k\in\K _C$ for any $k\in\K _{\alpha C}$, 
$\alpha\in\bigl(0,\frac{1}{\nu}\bigr)$, and
\begin{align*}
&(\hat{L}_0^*k)(\eta^+,\eta^-)\\=&  \sum_{y\in\eta^-}\int_{\Ga_0^2}
d\la^2(\xi^+,\xi^-) k(\xi^+\cup\eta^+\cup
y,\xi^-\cup\eta^-\setminus y)
A^+\bigl(y,\eta^+,\eta^-\setminus
y,\xi^+,\xi^-\bigr)\nonumber\\&-  \sum_{x\in\eta^+}\int_{\Ga_0^2}
d\la^2(\xi^+,\xi^-)
k(\xi^+\cup\eta^+,\xi^-\cup\eta^-) A^+\bigl(x,\eta^+\setminus
x,\eta^-,\xi^+,\xi^-\bigr)\nonumber\\&+
\sum_{x\in\eta^+}\int_{\Ga_0^2} d\la^2(\xi^+,\xi^-)\,k(\xi^+\cup\eta^+\setminus x,\xi^-\cup\eta^-\cup
x) A^-\bigl(x,\eta^+\setminus
x,\eta^-,\xi^+,\xi^-\bigr)\nonumber\\&-
\sum_{y\in\eta^-}\int_{\Ga_0^2} d\la^2(\xi^+,\xi^-)
k(\xi^+\cup\eta^+,\xi^-\cup\eta^-)
A^-\bigl(y,\eta^+,\eta^-\setminus y,\xi^+,\xi^-\bigr),\nonumber
\end{align*}
for $\lambda^2$-a.a.~$(\eta^+,\eta^-)\in\Ga_0^2$.
\end{proposition}

\section{Examples of rates}

For one-component systems there are many examples of
birth-and-death dynamics (e.g.~Glauber-type dynamics in mathematical
physics, Bolker--Dieckmann--Law--Pacala dynamics in mathematical
biology) as well as of hopping dynamics (e.g.~Kawasaki-type dynamics). 
These dynamics have been studied, in particular, in
\cite{FK2009,FKK09b,FKL06,KKL08, KoKtZh06,KLRII05,LK03}.

From the point of view of applications, multicomponent systems lead naturally 
to a richer situation due to many different possibilities for
concrete models and corresponding rates $b^\pm,d^\pm,c_i^\pm$, discussed in the 
previous sections. For instance, one may consider (birth-and-death) 
predator-prey models in which the death rate of preys (representing e.g.~the 
$+$-system) is higher due to the presence of a higher number of predators 
(representing the $-$-system) in a close neighborhood, while 
the birth rate of predators is higher if there is a higher number of preys 
nearby. For simplicity, assuming that there is no competition between
predators as well as between preys, typical rates are of the type
\begin{equation}\label{PP}
\begin{aligned}
d^+(x,\ga^+,\ga^-)&=m^++\sum_{y\in\ga^-}a_1(x-y), \\
d^-(y,\ga^+,\ga^-)&\equiv m^-,\\
b^+(x,\ga^+,\ga^-)&=\sum_{x'\in\ga^+}a_2(x-x'),\\
b^-(y,\ga^+,\ga^-)&=\sum_{y'\in\ga^-}a_3(y-y')\left(\kappa +
\sum_{x\in\ga^+}a_4(x-y')\right),
\end{aligned}
\end{equation}
for $m^\pm,\kappa>0$ and for even functions $0\leq a_i\in L^1(\X, dx)$, 
$i=1,2,3,4$. A similar situation occurs in other biological systems such as 
host-parasite or age-structured dynamics. On the other hand, on 
mathematical physics models, variants of the continuous Ising model 
\cite{GH96,GMRZ06,KZ07} (an analog of the Glauber dynamics) concern
birth and death rates of a different type. The simplest variant is
$d^\pm(x,\ga^+,\ga^-)\equiv m^\pm>0$ and
\begin{equation}\label{SI}
b^\pm(x,\ga^+,\ga^-)=b^\pm(x,\ga^\mp)=\exp\left(
-\sum_{y\in\ga^\mp}\phi(x-y) \right),
\end{equation}
with $\phi:\X\rightarrow\R\cup\{\infty\}$ being a pair-potential in $\X$.

These examples of rates are natural and quite general. Indeed, applications
deal with rates which are either ``linear'' functions
\[
\langle a_x,\ga^\pm\rangle:=\sum_{y\in\ga^\pm}a_x(y),
\]
with $a_x(y)=a(x-y)$ for some even function $a$, products of such linear 
functions on different variables $\ga^+,\ga^-$ (in particular, of polynomial 
type), or exponentials of these linear functions. For instance, in biological 
models concerning the so-called establishment and fecundity, rates are
naturally defined by products or superpositions of linear functions and
their exponentials (for the one-component case see \cite{FKK11b}). 

The results of the previous sections have shown that to derive explicit 
expressions for the mappings $\hat{L}$, $\hat{L}^*$ and to define sufficient
conditions allowing an extension of $\hat{L}$, $\hat{L}^*$ to linear operators 
one only has to study $A^\pm,B^\pm,C^\pm,D^\pm$. We explain now how to proceed 
for linear and exponential rates.

Let $b^\pm$, $d^\pm$ be defined as in \eqref{PP}. Then, for example for $d^+$,
\[
d^+(x,\eta^+\cup\ga^+,\eta^-\cup\ga^-)=m^++\sum_{y\in\eta^-}a_1(x-y)+\sum_{y\in\ga^-}a_1(x-y).
\]
By definitions \eqref{Ddef} of $D^+$ and \eqref{inverseK2} of $\KK^{-1}$, a 
simple calculation yields
\begin{align*}
D^+(x,\eta^+,\eta^-,\xi^+,\xi^-)=&\,\Bigl( m^++\sum_{y\in\eta^-}a_1(x-y)\Bigr)0^{|\xi^+|}0^{|\xi^-|}\\&+0^{|\xi^+|}\1_{\{\xi^-=\{y\}\}}a_1(x-y),
\end{align*}
being easy to show that for each $C>0$,
\[
\sum_{x\in\eta^+}\bigl\|D^+(x,\eta^+\setminus
x,\eta^-,\cdot,\cdot)\bigr\|_{\L_C}\leq
m|\eta^+|+\sum_{x\in\eta^+}\sum_{y\in\eta^-}a_1(x-y)+C|\eta^+|\int_\X dx\,a_1(x).
\]
Similar estimates naturally hold for $d^-$ and $b^\pm$. All together, 
these estimates yield an explicit form for the function $N$ introduced in 
\eqref{bdd}.

Let us now assume that $b^\pm$ are defined as in \eqref{SI} with $d^\pm$ being 
constants. Then,
\[
b^+(x,\eta^+\cup\ga^+,\eta^-\cup\ga^-)=\exp\left(
-\sum_{y\in\eta^-}\phi(x-y) \right)\exp\left(
-\sum_{y\in\ga^-}\phi(x-y) \right),
\]
and again the use of definitions \eqref{Ddef} and \eqref{inverseK2} leads to
\[
B^+(x,\eta^+,\eta^-,\xi^+,\xi^-)=0^{|\xi^+|}\exp\left(
-\sum_{y\in\eta^-}\phi(x-y) \right)\prod_{y\in\xi^-}\left(e^{-\phi(x-y)}-1\right).
\]
Assuming that $\phi(x)\geq- \upsilon $, $x\in\X$, for some $\upsilon \geq0$,
and $\beta:=\int_\X dx\,\bigl|e^{-\phi(x)}-1\bigr|<\infty$, we then obtain
\[
\sum_{x\in\eta^+}\bigl\|B^+(x,\eta^+\setminus x,\eta^-,\cdot,\cdot)\bigr\|_{\L_C}\leq|\eta^+|e^{\upsilon
|\eta^-|}e^{C\beta},
\]
where we have used the following equality which follows from definition 
\eqref{LPm} of the measure $\lambda$,
\[
\int_{\Ga_0}d\la(\xi^-)\prod_{y\in\xi^-}|f(y)|=\exp\left(\|f\|_{L^1(\X,dx)}\right),\quad f\in L^1(\X,dx).
\]
Similar estimates naturally hold for $b^-$, allowing at the end to derive an 
explicit form for the function $N$, introduced in \eqref{bdd}.

\subsection*{Acknowledgments}

Financial support of DFG through SFB 701 (Bielefeld University), 
German-Ukrainian Project 436 UKR 113/97 and FCT through
PTDC/MAT/100983/2008 and ISFL-1-209 are gratefully acknowledged.


\end{document}